\documentclass[showpacs,aps,twocolumn]{revtex4-1}
\usepackage{amssymb,amsmath,amsfonts,amsthm}
\usepackage{psfrag}
\usepackage{graphicx}
\newtheorem{thm}{Theorem}
\DeclareMathOperator*{\tr}{Tr}

\renewcommand{\>}{\rangle}

\begin{document}

\title{Extractable work from ensembles of quantum batteries. Entanglement helps.}

\date{\today}

\author{Robert~Alicki}
\email{fizra@univ.gda.pl}
\affiliation{Institute of Theoretical Physics and Astrophysics, University of Gda\'nsk, Poland}
\author{Mark~Fannes}
\email{mark.fannes@fys.kuleuven.be}
\affiliation{Instituut voor Theoretische Fysica, Universiteit Leuven, B-3001 Leuven, Belgium}

\begin{abstract}
Motivated by the recent interest in thermodynamics of micro- and mesoscopic quantum systems we study the maximal amount of work that can be reversibly extracted from a quantum system used to store temporarily energy. Guided by the notion of passivity of a quantum state we show that entangling unitary controls extract in general more work than independent ones. In the limit of large number of copies one can reach the thermodynamical bound given by the variational principle for free energy.
\end{abstract}


\maketitle

\section{Motivation}

The recent interest in models of quantum engines and refrigerators stimulates  theoretical efforts to precisely formulate fundamental thermodynamical principles and bounds valid on the micro- and nano-scale. In principle these can differ from the standard ones and converge to them only in the limit of macroscopic systems. A sample of references, including both general considerations and particular models, is given in~\cite{Qmachines}.

This paper is about the amount of work that can be extracted from a small quantum mechanical system that is used to store temporarily energy to transfer it from a production to a consumption center. To do so we are not coupling such a quantum battery to external thermal baths in order to drive thermodynamical engines but we address it by controlling its dynamics by external time-dependent fields. The battery comes with its initial state $\rho$ and own internal Hamiltonian $H$. The idealized process of reversible energy extraction is then governed by the system dynamics plus some fields that are only turned on during a certain interval $[0,\tau]$ of time. This leads to a time-dependent unitary dynamics of the battery. We now wonder about the maximal amount of work that can be extracted by such a process.

It has been known for a long time that some states can't deliver work in this way. Such states are called passive \cite{Pusz:1977}, \cite{Lenard:1978}. The maximal amount of work extractable from a battery is then the surplus energy of the initial state with respect to the passive state $\sigma_\rho$ with the same eigenvalues as $\rho$.

As we are dealing with small quantum systems we may wonder whether using processes that entangle two identical copies of a given battery can yield a higher energy extraction. More generally, what happens to a large number of copies?

We numerically demonstrate that the efficiency of energy extraction grows with the number of copies. Next we show rigorously that the maximal amount of extractable energy per battery asymptotically equals the energy difference between the initial state $\rho$ of the battery and the energy of the Gibbs state $\omega_{\overline\beta}$ with the same entropy as $\rho$. We indicate how to construct in principle a unitary that achieves this optimal bound.

\section{General context}

The Hilbert space $\mathcal H$ of wave functions of the battery is for simplicity chosen to be $d$-dimensional and we pick as standard basis for $\mathcal H$ the eigenvectors of the system Hamiltonian
\begin{equation}
H = \sum_{j=1}^d \epsilon_j\, |j\>\<j| \enskip\text{with}\enskip \epsilon_{j+1} > \epsilon_{j}.
\end{equation}
We assume here that the energy levels are non-degenerate which holds for a generic Hamiltonian.

The time-dependent fields that will be used to extract energy from the battery are described by $V(t) = V^\dagger(t)$ where $V(t)$ is possibly only different from zero for $0 \le t \le \tau$.

The initial state of the battery is described by a density matrix $\rho$ and the time evolution of $\rho$ is obtained from the Liouville-von~Neumann equation
\begin{equation}
\frac{d\ }{dt}\, \rho(t) = -i [H + V(t), \rho(t)],\enskip \rho(0) = \rho.
\end{equation}
The work extracted by this procedure is then
\begin{equation}
W = \tr \bigl( \rho H \bigr) - \tr \bigl( \rho(\tau) H \bigr)
\end{equation}
where the state at time $\tau$ is related to the initial state $\rho$ by a unitary transformation
\begin{equation}
\rho(\tau) = U(\tau)\, \rho\, U^\dagger(\tau)
\end{equation}
with $U(t)$ the time-ordered exponential of the total Hamiltonian $H + V(t)$:
\begin{equation}
U(\tau) = \text{Texp} \Bigl( -i \int_0^\tau \!ds\, \bigl( H + V(s) \bigr) \Bigr).
\end{equation}
Note that by a proper choice of controlling term $V$ any unitary $U$ can be obtained for $U(\tau)$. Therefore the maximal amount of extractable work (called \emph{ergotropy} in \cite{Allahverdyan:2004}) can be defined as
\begin{equation}
W_{\text{max}} := \tr \bigl( \rho H \bigr) - \min \tr \bigl( U\, \rho\, U^\dagger H \bigr)
\end{equation}
where the minimum is taken over all unitary transformations of $\mathcal H$.

Following Pusz and Woronowicz~\cite{Pusz:1977} and Lenard~\cite{Lenard:1978}, we call a state $\sigma$ passive if no work can be extracted from $\sigma$, i.e.\ if for all unitaries $U$
\begin{equation*}
\tr \bigl( \sigma H \bigr) \le \tr \bigl( U\, \sigma\, U^\dagger H \bigr).
\end{equation*}
The following theorem then holds:

\begin{thm}[\cite{Pusz:1977,Lenard:1978}]
$\sigma$ is passive if and only if
\begin{equation}
\sigma = \sum_{j=1}^d s_j\, |j\>\<j| \enskip\text{with}\enskip s_{j+1} \le s_{j}
\end{equation}
\end{thm}

In other words, $\sigma$ is passive if and only if it commutes with the system Hamiltonian and its eigenvalues are non-increasing with the energy. Given $\rho$ there is a unique passive state $\sigma_\rho$ minimizing $\tr U\, \rho\, U^\dagger H$. This state is obtained by a unitary rotation of $\rho$ denoted by $U_{\rho}$ and has the form
\begin{equation}
\sigma_\rho = U_{\rho}\,\rho\,U_{\rho}^{\dagger} = \sum_{j=1}^d r_j\, |j\>\<j|
\end{equation}
where $\{r_j\}$ are the eigenvalues of $\rho$ arranged in non-increasing order: $r_{j+1} \le r_{j}$. The corresponding minimal energy is $\sum_{j=1}^d r_j \epsilon_j$ and the maximal amount of extractable work is given by
\begin{equation}
W_{\text{max}} := \tr \bigl( \rho H \bigr) - \tr \bigl( \sigma_\rho\, H \bigr).
\label{wmax}
\end{equation}

\section{A general bound on available work}

We obtain here a bound on $W_{\text{max}}$ by comparing the energies of the passive state $\sigma_\rho$ and of the canonical Gibbs state $\omega_{\overline\beta}$ with the same entropy as $\rho$. Recall that the canonical Gibbs state at inverse temperature $\beta$ is given by
\begin{equation}
\omega_\beta = \frac{\exp(-\beta H)}{\mathcal Z}
\end{equation}
and that its von~Neumann entropy is strictly monotonically decreasing in $\beta$ with range $[0, \log d]$. The von~Neumann entropy $S(\rho)$ of a density matrix $\rho$ is
\begin{equation}
S(\rho) = - \tr \rho \log\rho.
\end{equation}
For any given density matrix $\rho$ on $\mathcal H$ there exists therefore a unique inverse temperature $\overline\beta$  such that $S(\rho) = S(\omega_{\overline\beta})$. The relation between $\rho$ and $\overline\beta$ is, of course, highly non-linear.

We now use the variational principle of statistical mechanics that asserts that the Gibbs canonical density matrix is that which minimizes the free energy:
\begin{equation}
\tr \bigl( \rho H \bigr)- {\overline\beta}^{-1}\, S(\rho) \ge \tr \bigl( \omega_{\overline\beta}\, H \bigr) - {\overline\beta}^{-1}\, S(\omega_{\overline\beta}).
\label{var}
\end{equation}
With our choice of $\overline\beta$ we obtain that
\begin{equation}
\tr \bigl( \rho H \bigr) \ge \tr \bigl( \sigma_{\rho} H \bigr)\ge\tr \bigl( \omega_{\overline\beta}\, H \bigr).
\label{var1}
\end{equation}
and hence the thermodynamical bound on the available work
\begin{equation}
W_{\text{max}} \le \tr \bigl( \rho\, H \bigr) - \tr \bigl(\omega_{\overline\beta}\, H \bigr).
\label{bound}
\end{equation}
Generally, $\omega_{\overline\beta}$ is different from $\sigma_\rho$ as $\omega_{\overline\beta}$ and $\sigma_\rho$ or $\rho$ have different eigenvalues. Note, however, that the two-dimensional case is exceptional because there is a one-to-one correspondence between the entropy of a qubit state and its ordered eigenvalues. Generally it is not true that a product of two independent copies of a passive state still is passive. There is therefore a possibility of improving over~(\ref{wmax}) on the amount of extractable work per copy for several copies of a system. In other words, by using entangling unitaries, one can in principle beat~(\ref{wmax}).

\psfrag{b}{\hspace{-2pt}\raisebox{-3pt}{10}}
\psfrag{c}{\hspace{-2pt}\raisebox{-3pt}{20}}
\psfrag{d}{\hspace{-2pt}\raisebox{-3pt}{30}}
\psfrag{e}{\hspace{-2pt}\raisebox{-3pt}{40}}
\psfrag{n}{$n$}
\psfrag{w}{$e^{(n)}$}
\psfrag{r}{\hspace{-12pt}\raisebox{3pt}{$e^{(\infty)}$}}
\psfrag{g}{\hspace{-9pt}\raisebox{3pt}{$e^{(1)}$}}
\begin{figure}[h]
\begin{center}
\includegraphics[width=.45 \textwidth]{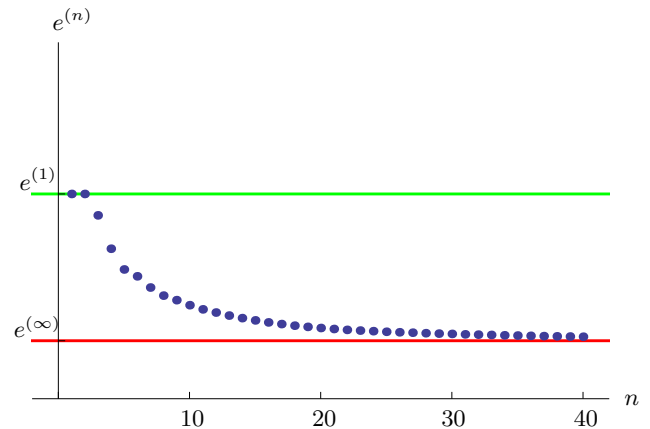}
\end{center}
\caption{Energy per copy of passive state $\sigma_{\otimes^n \rho}$ associated to $\otimes^n \rho$}
\end{figure}

In Fig.1 the energies $e^{(n)}$ per copy of the passive state $\sigma_{\otimes^n \rho}$ obtained from a product state $\otimes^n \rho$ are plotted as dots for $n=1, 2, \ldots, 40$. The lower line shows the asymptotic value of $e^{(n)}$. The system is a three level battery with energy levels $\{0,0.579,1\}$ and the passive state corresponding to the initial density matrix has eigenvalues $\{0.538,0.237,0.224\}$. The values $e^{(n)}$ have been obtained by rearranging the eigenvalues of $\otimes^n \rho$ and the $n$-copy Hamiltonian $H^{(n)}$, see~(\ref{hamn}). The maximal additional work that can be extracted on top of the single copy extractable work using entangling unitaries is the difference between $e^{(1)}$ and $e^{(\infty)}$. We will compute this value in the next section.

\section{Entangling batteries}

A state $\sigma$ is called completely passive if $\otimes^n \sigma$ is passive for all $n = 1, 2, \ldots$  with respect to the sum Hamiltonian
\begin{equation}
H^{(n)} = \sum_{j=1}^n H_j
\label{hamn}
\end{equation}
where $H_j$ is the $j$-th independent copy of $H$. Thermodynamic equilibrium is equivalent to complete passivity:

\begin{thm}[\cite{Pusz:1977,Lenard:1978}]
$\sigma$ is completely passive if and only if it is a Gibbs state.
\end{thm}

We now consider $n$ independent copies of our battery and apply the general bound~(\ref{bound}) to estimate the maximal amount of available work per battery
\begin{align}
w^n_{\text{max}}
&:= \frac{1}{n}\, \Bigl\{ \tr \bigl( (\otimes^n \rho)\, H^{(n)} \bigr) - \tr \bigl( \sigma_{\otimes^n \rho}\, H^{(n)} \bigr) \Bigr\} \nonumber \\
&\le \tr \bigl( \rho H \bigr) - \tr \bigl( \omega_{\overline\beta}\, H \bigr).
\label{optimal}
\end{align}
It is our aim to show that this bound is actually asymptotically achievable:

\begin{thm}

\begin{equation}
\lim_{n\to\infty} w^n_{\text{max}}= \tr \bigl( \rho H \bigr) - \tr \bigl( \omega_{\overline\beta}\, H \bigr).
\label{ineq}
\end{equation}
\end{thm}

\begin{proof}
The proof is based on the idea of typical configurations. Assume that $\rho$ is diagonal in the eigenbasis of $H$, this can always be achieved by a suitable unitary rotation and let $\{r_j\}$ be the eigenvalues of $\rho$ arranged in non-decreasing order. A configuration of length $n$ is an $n$-tuple $|\mathbf i \> = |i_1, i_2, \ldots, i_n\>$ of indices in $\{1,2,\ldots, d\}$ with corresponding eigenvalue $r_{i_1} r_{i_2} \cdots r_{i_n}$ of $\otimes^n \rho$. Typical configurations will be of the type $|\mathbf i\>$ where the number of times the index $k \in \{1,2,\ldots,d\}$ occurs lies between $(r_k - \epsilon) n$ and $(r_k + \epsilon) n$ for every $k$. The subspace spanned by these typical $|\mathbf i\>$ has approximately dimension $\exp(n S(\rho))$ and each such configuration corresponds to the average energy
\begin{equation}
\frac{1}{n}\, H^{(n)}\, |\mathbf i\> = \Bigl( \sum_{j=1}^d r_j\, \epsilon_j \Bigr) |\mathbf i> + \text{o}(\epsilon).
\label{typical}
\end{equation}

Now we repeat the same construction for the product of $n$ copies of the Gibbs state $\omega_{\overline\beta}$. As $S(\omega_{\overline\beta}) = S(\rho)$, the typical subspaces of $\otimes^n \rho$ and $\otimes^n \omega_{\overline\beta}$ have approximately the same dimension. Moreover, for  both product states the probability of finding a system outside the typical subspaces is $\text{o}(\epsilon)$. We can now find a unitary $U^{(\epsilon)}$ on $\otimes^n \mathcal{H}$ that maps one subspace into the other. This unitary is highly non-unique, and generally differs from the optimal reordering given by $U_{\otimes^n\rho}$ but nevertheless produces the state with the energy close to the optimal one, i.e.
\begin{equation}
\bigl| \tr \bigl(U^{(\epsilon)} (\otimes^n \rho)\,{U^{(\epsilon)}}^{\dagger}\, H^{(n)} \bigr) -\tr \bigl( (\otimes^n \omega_{\overline\beta})\, H^{(n)} \bigr) \bigr| \le n\, \text{o}(\epsilon) .
\label{optimal1}
\end{equation}
Using \eqref{var1} we obtain
\begin{align}
\tr \bigl( \otimes^n \omega_{\overline\beta}\, H^{(n)} \bigr) &\leq \tr \bigl( \sigma_{\otimes^n \rho}\, H^{(n)} \bigr)\\
\nonumber
&\le \tr \bigl(U^{(\epsilon)} (\otimes^n \rho)\,{U^{(\epsilon)}}^{\dagger}\, H^{(n)} \bigr),
\label{optimal2}
\end{align}
which combined with~\eqref{optimal1} yields the final estimation
\begin{align}
&\tr \bigl( \rho H \bigr) - \tr \bigl( \omega_{\overline\beta}\, H \bigr)\geq w^n_{\text{max}} \geq\\
\nonumber
&\tr \bigl( \rho H \bigr) - \tr \bigl( \omega_{\overline\beta}\, H \bigr)- \text{o}(\epsilon).
\end{align}
\end{proof}

Remark that a unitary transforming $\otimes^n \rho$ into $\sigma_{\otimes^n \rho}$ cannot be product and must therefore dynamically entangle the $n$ batteries. In the numerical example of Fig.~1 the asymptotic value $e^{(\infty)}$ exactly coincides with $\omega_{\overline\beta}$.

\section{Conclusion}

The notion of maximal reversibly extractable work for a quantum battery motivated by the concept of passivity is discussed. It is applicable to  full quantum models of micro- or mesoscopic machines where work is supplied or extracted by a quantum system (`quantum battery', `work reservoir') instead of a time-dependent perturbation of the Hamiltonian. A proper definition of work is important to develop a consistent thermodynamics of small quantum systems which is relevant in nanotechnology and biophysics. Generally, the extractable work is smaller than the thermodynamical bound computed using variational principle for a free energy. Using entanglement one can in general extract more work per battery from several independent copies of a battery and asymptotically  reach the thermodynamical bound. However, the optimal procedures of work extraction are generally difficult to implement by realistic control Hamiltonians. An interesting problem for future investigation is to find efficiency bounds when practical restrictions are imposed on the available control mechanisms.

\begin{acknowledgments}
R.A. acknowledges the support by the Polish Ministry of Science and Higher Education, grant NN 202208238 and M.F. the FWO Vlaanderen project G040710N.
\end{acknowledgments}

\end{document}